\documentclass[aps, pra, reprint, superscriptaddress, onecolumn, notitlepage, tightenlines, 11pt]{revtex4-2}

\pdfoutput=1

\usepackage{header}

\graphicspath{{graphics/}}

\begin{document}

\bibliographystyle{apsrev4-2}

\title{Approximate Counting in Local Lemma Regimes}

\author{Ryan L. Mann}
\email{mail@ryanmann.org}
\homepage{http://www.ryanmann.org}
\affiliation{Centre for Quantum Computation and Communication Technology, Centre for Quantum Software and Information, School of Computer Science, Faculty of Engineering \& Information Technology, University of Technology Sydney, NSW 2007, Australia}

\author{Gabriel Waite}
\affiliation{Centre for Quantum Computation and Communication Technology, Centre for Quantum Software and Information, School of Computer Science, Faculty of Engineering \& Information Technology, University of Technology Sydney, NSW 2007, Australia}

\begin{abstract}
    We establish efficient approximate counting algorithms for several natural problems in local lemma regimes. In particular, we consider the probability of intersection of events and the dimension of intersection of subspaces. Our approach is based on the cluster expansion method. We obtain fully polynomial-time approximation schemes for both the probability of intersection and the dimension of intersection for commuting projectors. For general projectors, we provide two algorithms: a fully polynomial-time approximation scheme under a global inclusion-exclusion stability condition, and an efficient affine approximation under a spectral gap assumption. As corollaries of our results, we obtain efficient algorithms for approximating the number of satisfying assignments of conjunctive normal form formulae and the dimension of satisfying subspaces of quantum satisfiability formulae.
\end{abstract}

\maketitle

{
\hypersetup{linkcolor=black}
\tableofcontents
}

\newpage

\section{Introduction}
\label{section:Introduction}

Approximate counting is a fundamental problem in computational complexity and algorithms. The goal is to efficiently approximate the number of elements in a combinatorially defined set. This problem arises naturally across mathematics, computer science, and statistical physics in contexts such as computing probabilities in discrete probability spaces, counting satisfying assignments of constraint satisfaction formulae, and evaluating partition functions of statistical mechanical models. Exact counting is typically \#\textsc{P}-hard for these problems. Nevertheless, efficient approximation algorithms exist for many natural counting problems through approaches including Markov-chain Monte Carlo~\cite{jerrum1989approximating}, correlation decay~\cite{weitz2006counting}, and interpolation-type methods~\cite{barvinok2016combinatorics, helmuth2020algorithmic}.

The Lov\'asz local lemma~\cite{erdos1975problems, shearer1985problem} provides a powerful tool for establishing the existence of combinatorial objects. Informally, the lemma states that if a collection of events each occurs with sufficiently small probability and each event depends on only a bounded number of other events, then with positive probability none of the events occur. The quantum Lov\'asz local lemma~\cite{ambainis2012quantum, sattath2016local, he2019quantum} extends this principle to the quantum setting, providing conditions under which the intersection of subspaces has positive dimension. Both the classical and quantum versions guarantee existence under conditions of weak dependencies and appropriate bounds on local parameters. A central question is the development of efficient algorithms in such local lemma regimes.

The algorithmic study of local lemma regimes encompasses three complementary problems: constructing explicit solutions, approximately counting solutions, and approximately sampling solutions. Efficient constructive algorithms exist in the classical setting~\cite{moser2010constructive}, in the commuting quantum setting~\cite{arad2013constructive, schwarz2013information}, and in the general quantum setting under the assumption of a uniform spectral gap~\cite{gilyen2017preparing}. For approximate counting, Barvinok~\cite{barvinok2025computing} established a quasi-polynomial time algorithm in the classical setting. For approximate sampling, polynomial-time algorithms have been developed for structured classical cases~\cite{moitra2019approximate, guo2019uniform, feng2021sampling, jain2022towards, jain2021sampling, he2021perfect, he2022sampling, feng2023towards, he2023deterministic, wang2024sampling}. However, the development of efficient approximate counting and sampling algorithms in the quantum Lov\'asz local lemma regime remains open.

In this paper, we establish efficient approximate counting algorithms for several natural problems in local lemma regimes. Our main results are as follows. First, we provide a fully polynomial-time approximation scheme for the probability of intersection of events. Second, we establish a fully polynomial-time approximation scheme for the dimension of intersection for commuting projectors. Third, we provide two algorithms for the dimension of intersection for general projectors: a fully polynomial-time approximation scheme under a global inclusion-exclusion stability condition, and an affine approximation under a spectral gap assumption. Our results and their relationship to prior work on existence, constructive, counting, and sampling algorithms in local lemma regimes are summarised in Table~\ref{table:summaryresultslocallemmaregimes}. As corollaries of our results, we obtain efficient algorithms for approximating the number of satisfying assignments of conjunctive normal form (CNF) formulae and the dimension of satisfying subspaces of quantum satisfiability formulae.

\begin{table}[ht]
    \centering
    \setcellgapes{3pt}
    \makegapedcells
    \begin{tabularx}{\textwidth}{XXXXX}
        \hline
        \makecell{\textbf{Problem}} & \makecell{\textbf{Existence}} & \makecell{\textbf{Constructive}} & \makecell{\textbf{Counting}} & \makecell{\textbf{Sampling}} \\
        \hline
        \makecell{Probability of \\ Intersection} & \makecell{Yes~\cite{erdos1975problems, shearer1985problem}} & \makecell{Yes~\cite{moser2010constructive}} & \makecell{Yes~\cite{barvinok2025computing} \\ Theorem~\ref{theorem:ApproximationAlgorithmProbabilityIntersection}} & \makecell{Partial~\cite{moitra2019approximate, guo2019uniform, feng2021sampling, jain2022towards, jain2021sampling, he2021perfect, he2022sampling, feng2023towards, he2023deterministic, wang2024sampling}} \\
        \makecell{Commuting \\ Dimension of \\ Intersection} & \makecell{Yes~\cite{ambainis2012quantum, sattath2016local, he2019quantum}} & \makecell{Yes~\cite{arad2013constructive, schwarz2013information}} & \makecell{Yes \\ Theorem~\ref{theorem:ApproximationAlgorithmDimensionIntersectionCommuting}} & \makecell{Open} \\
        \makecell{General \\ Dimension of \\ Intersection} & \makecell{Yes~\cite{ambainis2012quantum, sattath2016local, he2019quantum}} & \makecell{Partial~\cite{gilyen2017preparing}} & \makecell{Partial \\ Theorems~\ref{theorem:ApproximationAlgorithmDimensionIntersectionGeneral} and~\ref{theorem:DetectabilityApproximationAlgorithmDimensionIntersectionGeneral}} & \makecell{Open} \\
        \hline
    \end{tabularx}
    \caption{Summary of results in local lemma regimes.}
    \label{table:summaryresultslocallemmaregimes}
\end{table}

Our algorithms are based on the cluster expansion from mathematical physics~\cite{friedli2017statistical}. We formulate each counting problem as an abstract polymer model in the formalism of Koteck\'y and Preiss~\cite{kotecky1986cluster}. When the polymer weights satisfy a suitable decay condition, the cluster expansion provides a convergent power series representation for the logarithm of the polymer model partition function. Under additional assumptions on the polymer structure, this yields efficient approximation algorithms via truncation~\cite{helmuth2020algorithmic, borgs2020efficient}. This method has been used to design efficient approximate counting algorithms for classical partition functions, such as the hardcore model~\cite{helmuth2020algorithmic, jenssen2020algorithms, chen2019fast, cannon2020counting, jenssen2020independent, jenssen2023approximately, galvin2024zeroes, collares2025counting} and the Potts model~\cite{helmuth2020algorithmic, chen2019fast, borgs2020efficient, helmuth2023finite, carlson2024algorithms}, for quantum partition functions~\cite{mann2021efficient, helmuth2023efficient, mann2024algorithmic}, and for volumes of graph-theoretic polytopes~\cite{guo2024deterministic}. We refer the reader to Refs.~\cite{patel2022approximate, jenssen2024cluster} for surveys on this topic.

This paper is structured as follows. In Section~\ref{section:Preliminaries}, we introduce the necessary preliminaries. Then, in Section~\ref{section:ApproximationAlgorithms}, we establish an algorithm for approximating abstract polymer model partition functions. In Section~\ref{section:Applications}, we apply this algorithm to establish our main counting results in local lemma regimes. Finally, we conclude in Section~\ref{section:ConclusionAndOutlook} with some remarks and open problems.

\section{Preliminaries}
\label{section:Preliminaries}

\subsection{Graph Theory}
\label{section:GraphTheory}

Let $G$ be a graph with vertex set $V(G)$ and edge set $E(G)$. We denote the \emph{order} of $G$ by $\abs{G}=\abs{V(G)}$. The \emph{maximum degree} $\Delta(G)$ of $G$ is the maximum degree over all vertices of $G$. For a subset $S \subseteq V(G)$, the \emph{induced subgraph} $G[S]$ is the subgraph of $G$ whose vertex set is $S$ and whose edge set consists of all edges in $G$ which have both endpoints in $S$. An \emph{independent set} of $G$ is a subset $S$ of $V(G)$ with no edges between them. A \emph{proper colouring} of $G$ is a partition of $V(G)$ into independent sets. The \emph{chromatic number} $\chi(G)$ of $G$ is the minimum number of parts over all proper colourings of $G$. The chromatic number satisfies the bound $\chi(G)\leq\Delta(G)+1$. This can be improved to $\chi(G)\leq\Delta(G)$ if $G$ is a connected graph that is neither a complete graph nor an odd cycle~\cite{brooks1941colouring}. We denote the \emph{complete graph} on $n$ vertices by $K_n$. For two graphs $G$ and $H$, the \emph{strong product} $G \boxtimes H$ is the graph whose vertex set is $V(G) \times V(H)$ and whose edge set consists of all pairs of distinct vertices $\{(u,u'),(v,v')\}$ such that either $u=v$ and $\{u',v'\} \in E(H)$, $u'=v'$ and $\{u,v\} \in E(G)$, or $\{u,v\} \in E(G)$ and $\{u',v'\} \in E(H)$.

\subsection{Abstract Polymer Models}
\label{section:AbstractPolymerModels}

An \emph{abstract polymer model} is a triple $(\mathcal{C},w,\sim)$, where $\mathcal{C}$ is a countable set whose elements are called \emph{polymers}, $w:\mathcal{C}\to\mathbb{C}$ is a function that assigns to each polymer a complex number $w_\gamma$ called the \emph{weight} of the polymer, and $\sim$ is a \emph{symmetric compatibility relation} such that each polymer is incompatible with itself. A set of polymers is called \emph{admissible} if the polymers in the set are all pairwise compatible. Note that the empty set is admissible. We denote by $\mathcal{G}$ the collection of all admissible sets of polymers from $\mathcal{C}$. The abstract polymer model partition function is defined by
\begin{equation}
    Z(\mathcal{C},w) \coloneqq \sum_{\Gamma\in\mathcal{G}}\prod_{\gamma\in\Gamma}w_\gamma. \notag
\end{equation}

The abstract polymer model formulation offers a convenient way of applying cluster expansion techniques to counting problems.

\subsection{Abstract Cluster Expansion}
\label{section:AbstractClusterExpansion}

We now define the \emph{abstract cluster expansion}~\cite{kotecky1986cluster, friedli2017statistical}. Let $\Gamma$ be a non-empty ordered tuple of polymers. The \emph{incompatibility graph} $H_\Gamma$ is the graph whose vertex set is $\Gamma$ and has an edge between vertices $\gamma$ and $\gamma'$ if and only if they are incompatible. We say that $\Gamma$ is a \emph{cluster} if its incompatibility graph $H_\Gamma$ is connected. We denote by $\mathcal{G}_C$ the set of all clusters of polymers from $\mathcal{C}$. The abstract cluster expansion is a formal power series for $\log(Z(\mathcal{C},w))$ in the variables $w_\gamma$, defined by
\begin{equation}
    \log(Z(\mathcal{C},w)) \coloneqq \sum_{\Gamma\in\mathcal{G}_C}\varphi(H_\Gamma)\prod_{\gamma\in\Gamma}w_\gamma, \notag
\end{equation}
where $\varphi(\;\cdot\;)$ denotes the \emph{Ursell function}. Specifically, for a graph $H$, $\varphi(H)$ is defined by
\begin{equation}
    \varphi(H) \coloneqq \frac{1}{\abs{H}!}\sum_{\substack{S \subseteq E(H) \\ \text{spanning} \\ \text{connected}}}(-1)^\abs{S}, \notag
\end{equation}
where the sum is over all spanning connected edge sets of $H$.

In the abstract setting, we consider a function $\abs{\;\cdot\;}:\mathcal{C}\to\mathbb{Z}^+$ that assigns to each polymer a positive integer $\abs{\gamma}$ called the \emph{size} of the polymer. In the sequel, we shall consider polymers that are connected induced subgraphs of a graph. For such polymers, the size of the polymer is taken to be its order, i.e., $\abs{\gamma}=\abs{V(\gamma)}$. For algorithmic applications, it is convenient to consider the truncated cluster expansion $T_m(Z(\mathcal{C},w))$ to order $m$, defined by
\begin{equation}
    T_m(Z(\mathcal{C},w)) \coloneqq \sum_{\substack{\Gamma\in\mathcal{G}_C \\ \abs{\Gamma} \leq m}}\varphi(H_\Gamma)\prod_{\gamma\in\Gamma}w_\gamma, \notag
\end{equation}
where $\abs{\Gamma}\coloneqq\sum_{\gamma\in\Gamma}\abs{\gamma}$.

\subsection{Approximation Schemes}
\label{section:ApproximationSchemes}

Let $\epsilon>0$ be a real number. A \emph{multiplicative $\epsilon$-approximation} to a number $\alpha$ is a number $\hat{\alpha}$ such that $\abs{\alpha-\hat{\alpha}}\leq\epsilon\abs{\alpha}$. A \emph{fully polynomial-time approximation scheme} for a sequence of numbers $(\alpha_n)_{n\in\mathbb{N}}$ is a deterministic algorithm that, for any $n$ and $\epsilon>0$, produces a multiplicative $\epsilon$-approximation to $\alpha_n$ in time polynomial in $n$ and $1/\epsilon$. A \emph{fully polynomial-time randomised approximation scheme} for a sequence of numbers $(\alpha_n)_{n\in\mathbb{N}}$ is a randomised algorithm that, for any $n$ and $\epsilon>0$, produces a multiplicative $\epsilon$-approximation to $\alpha_n$ with probability at least $2/3$ in time polynomial in $n$ and $1/\epsilon$.

\subsection{Complexity Theory}
\label{section:ComplexityTheory}

We shall refer to the following complexity classes: \textsc{P} (polynomial time), \textsc{RP} (randomised polynomial time), \textsc{NP} (non-deterministic polynomial time), and \#\textsc{P}. For a formal definition of these complexity classes, we refer the reader to Ref.~\cite{arora2009computational}.

\section{Approximation Algorithms}
\label{section:ApproximationAlgorithms}

In this section, we establish an efficient algorithm for approximating abstract polymer model partition functions. We consider abstract polymer models in which the polymers are connected induced subgraphs of bounded-degree graphs and compatibility is defined by disconnectedness. Under a suitable decay condition on the polymer weights, the logarithm of the partition function can be controlled by a convergent cluster expansion.

The algorithm approximates the partition function by evaluating the truncated cluster expansion to sufficiently high order. This algorithm is essentially due to Helmuth, Perkins, and Regts~\cite{helmuth2020algorithmic} and Borgs et al.~\cite{borgs2020efficient}, with slight modifications in the analysis. The result is summarised in the following theorem.

\begin{theorem}
    \label{theorem:ApproximationAlgorithmAbstractPolymerModelPartitionFunction}
    Fix $\delta>0$. Let $G$ be a graph of maximum degree at most $\Delta$. Further let $(\mathcal{C},w,\sim)$ be an abstract polymer model such that the polymers are connected induced subgraphs of $G$ and that two polymers $\gamma$ and $\gamma'$ are compatible if and only if $V(\gamma) \cap V(\gamma') = \varnothing$ and $G[V(\gamma) \cup V(\gamma')] = \gamma\cup\gamma'$. Suppose that, for all polymers $\gamma\in\mathcal{C}$, the weight $w_\gamma$ can be computed in time $\exp(O(\abs{\gamma}))$ and satisfies
    \begin{equation}
        \abs{w_\gamma} \leq \left(\frac{1}{e^{1+\delta}(2\Delta+1)}\right)^\abs{\gamma}. \notag
    \end{equation}
    Then the cluster expansion for $\log(Z(\mathcal{C},w))$ converges absolutely, $Z(\mathcal{C},w)\neq0$, and there is a fully polynomial-time approximation scheme for $Z(\mathcal{C},w)$.
\end{theorem}

In Section~\ref{section:Applications} we shall apply Theorem~\ref{theorem:ApproximationAlgorithmAbstractPolymerModelPartitionFunction} to establish efficient approximation algorithms for several counting problems in local lemma regimes.

The proof of Theorem~\ref{theorem:ApproximationAlgorithmAbstractPolymerModelPartitionFunction} requires several lemmas. We first show that provided the polymer weights satisfy a suitable decay condition, then the cluster expansion converges absolutely and the truncated cluster expansion provides a good approximation to $\log(Z(\mathcal{C},w))$. This is formalised by the following lemma which utilises the Koteck\'y-Preiss convergence criterion~\cite{kotecky1986cluster}.

\begin{lemma}
    \label{lemma:ConvergenceClusterExpansion}
    Fix $\delta>0$. Let $G$ be a graph of maximum degree at most $\Delta$. Further let $(\mathcal{C},w,\sim)$ be an abstract polymer model such that the polymers are connected induced subgraphs of $G$ and that two polymers $\gamma$ and $\gamma'$ are compatible if and only if $V(\gamma) \cap V(\gamma') = \varnothing$ and $G[V(\gamma) \cup V(\gamma')] = \gamma\cup\gamma'$. Suppose that, for all polymers $\gamma\in\mathcal{C}$, the weight $w_\gamma$ satisfies
    \begin{equation}
        \abs{w_\gamma} \leq \left(\frac{1}{e^{1+\delta}(2\Delta+1)}\right)^\abs{\gamma}. \notag
    \end{equation}
    Then the cluster expansion for $\log(Z(\mathcal{C},w))$ converges absolutely, $Z(\mathcal{C},w)\neq0$, and for $m\in\mathbb{Z}^+$,
    \begin{equation}
        \abs{T_m(Z(\mathcal{C},w))-\log(Z(\mathcal{C},w))} \leq \frac{\Delta+1}{2\Delta+1}\abs{G}e^{-\delta m}. \notag
    \end{equation}
\end{lemma}
\begin{proof}
    We introduce a polymer $\gamma_u$ to every vertex $u$ in $G$ consisting of only that vertex. The weight of $\gamma_u$ is defined to be $w_{\gamma_u}=0$. We define $\gamma_u$ to be incompatible with every polymer $\gamma$ such that either $V(\gamma_u) \cap V(\gamma) \neq \varnothing$ or $G[V(\gamma_u) \cup V(\gamma)] \neq \gamma_u\cup\gamma$. To apply the Koteck\'y-Preiss convergence criterion, we bound the following exponentially weighted sum of polymer weights incompatible with $\gamma_u$.
    \begin{equation}
        \sum_{\gamma\nsim\gamma_u}\abs{w_\gamma}e^{\left(\frac{\Delta+1}{2\Delta+1}+\delta\right)\abs{\gamma}} \leq \sum_{\gamma\nsim\gamma_u}\left(\frac{1}{e^{1+\delta}(2\Delta+1)}e^{\frac{\Delta+1}{2\Delta+1}+\delta}\right)^\abs{\gamma} = \sum_{\gamma\nsim\gamma_u}\left(\frac{1}{e^{\frac{\Delta}{2\Delta+1}}(2\Delta+1)}\right)^\abs{\gamma}. \notag
    \end{equation}
    For a vertex $v$, the number of polymers with $\abs{\gamma}=m$ that contain $v$ is at most $\frac{(m\Delta)^{m-1}}{m!}$~\cite[Lemma 2.1]{borgs2013left}. Hence, by a union bound, the number of polymers with $\abs{\gamma}=m$ that are incompatible with $\gamma_u$ is at most $(\Delta+1)\frac{(m\Delta)^{m-1}}{m!}$. Thus, we may write
    \begin{align}
        \sum_{\gamma\nsim\gamma_u}\abs{w_\gamma}e^{\left(\frac{\Delta+1}{2\Delta+1}+\delta\right)\abs{\gamma}} &\leq (\Delta+1)\sum_{m=1}^\infty\frac{(m\Delta)^{m-1}}{m!}\left(\frac{1}{e^{\frac{\Delta}{2\Delta+1}}(2\Delta+1)}\right)^m \notag \\
        &= \frac{(\Delta+1)}{\Delta}\sum_{m=1}^\infty\frac{m^{m-1}}{m!}\left(\frac{\Delta}{e^{\frac{\Delta}{2\Delta+1}}(2\Delta+1)}\right)^m \notag \\
        &= \frac{(\Delta+1)}{\Delta}T\left(\frac{\Delta}{e^{\frac{\Delta}{2\Delta+1}}(2\Delta+1)}\right), \notag
    \end{align}
    where $T(\;\cdot\;)$ denotes the \emph{tree function}. Specifically, $T(z)$ is the unique formal power series solution to the functional equation $T(z)=ze^{T(z)}$. We have
    \begin{equation}
        T\left(\frac{\Delta}{e^{\frac{\Delta}{2\Delta+1}}(2\Delta+1)}\right) = \frac{\Delta}{2\Delta+1}. \notag
    \end{equation}
    Therefore,
    \begin{equation}
        \sum_{\gamma\nsim\gamma_u}\abs{w_\gamma}e^{\left(\frac{\Delta+1}{2\Delta+1}+\delta\right)\abs{\gamma}} \leq \frac{\Delta+1}{2\Delta+1}. \notag
    \end{equation}
    Fix a polymer $\gamma^*$. By summing over all vertices in $\gamma^*$, we obtain
    \begin{equation}
        \sum_{\gamma\nsim\gamma^*}\abs{w_{\gamma}}e^{\left(\frac{\Delta+1}{2\Delta+1}+\delta\right)\abs{\gamma}} \leq \frac{\Delta+1}{2\Delta+1}\abs{\gamma^*}. \notag
    \end{equation}
    Now, by applying the main theorem of Ref.~\cite{kotecky1986cluster} with $a(\gamma)=\frac{\Delta+1}{2\Delta+1}\abs{\gamma}$ and $d(\gamma)=\delta\abs{\gamma}$, we have that the cluster expansion converges absolutely, $Z(\mathcal{C},w)\neq0$, and
    \begin{equation}
        \sum_{\substack{\Gamma\in\mathcal{G}_C \\ \Gamma\ni\gamma_u }}\abs{\varphi(H_\Gamma)\prod_{\gamma\in\Gamma}w_\gamma}e^{\delta\abs{\Gamma}} \leq \frac{\Delta+1}{2\Delta+1}. \notag
    \end{equation}
    By summing over all vertices in $G$, we obtain
    \begin{equation}
        \sum_{\substack{\Gamma\in\mathcal{G}_C \\ \abs{\Gamma} \geq m }}\abs{\varphi(H_\Gamma)\prod_{\gamma\in\Gamma}w_\gamma} \leq \frac{\Delta+1}{2\Delta+1}\abs{G}e^{-\delta m}, \notag
    \end{equation}
    completing the proof.
\end{proof}

Lemma~\ref{lemma:ConvergenceClusterExpansion} implies that to obtain a multiplicative $\epsilon$-approximation to $Z(\mathcal{C},w)$, it is sufficient to compute the truncated cluster expansion $T_m(Z(\mathcal{C},w))$ to order $m=O(\log(\abs{G}/\epsilon))$. We now establish an algorithm computing $T_m(Z(\mathcal{C},w))$ in time $\exp(O(m))\cdot\abs{G}^{O(1)}$, which requires the following two lemmas.

\begin{lemma}
    \label{lemma:ListClustersAlgorithm}
    Let $G$ be a graph of maximum degree at most $\Delta$. Further let $(\mathcal{C},w,\sim)$ be an abstract polymer model such that the polymers are connected induced subgraphs of $G$ and that two polymers $\gamma$ and $\gamma'$ are compatible if and only if $V(\gamma) \cap V(\gamma') = \varnothing$ and $G[V(\gamma) \cup V(\gamma')] = \gamma\cup\gamma'$. The clusters of size at most $m$ can be listed in time $\exp(O(m))\cdot\abs{G}^{O(1)}$.
\end{lemma}
\begin{proof}
    Our proof follows that of Ref.~\cite[Theorem 6]{helmuth2020algorithmic}. We enumerate all connected induced subgraphs of $G$ of order $m$ in time $\exp(O(m))\cdot\abs{G}^{O(1)}$ by Ref.~\cite[Lemma 3.4]{patel2017deterministic}. The clusters of size at most $m$ can then be listed in time $\exp(O(m))\cdot\abs{G}^{O(1)}$ as in the proof of Ref.~\cite[Theorem 6]{helmuth2020algorithmic}.
\end{proof}

\begin{lemma}
    \label{lemma:UrsellFunctionAlgorithm}
    The Ursell function $\varphi(H)$ can be computed in time $\exp(O(\abs{H}))$.
\end{lemma}
\begin{proof}
    This is a result of Ref.~\cite{bjorklund2008computing}; see Ref.~\cite[Lemma 5]{helmuth2020algorithmic}.
\end{proof}

\begin{lemma}
    \label{lemma:TruncatedClusterExpansionApproximationAlgorithm}
    Fix $\Delta\in\mathbb{Z}_{\geq2}$. Let $G$ be a graph of maximum degree at most $\Delta$. Further let $(\mathcal{C},w,\sim)$ be an abstract polymer model such that the polymers are connected induced subgraphs of $G$ and that two polymers $\gamma$ and $\gamma'$ are compatible if and only if $V(\gamma) \cap V(\gamma') = \varnothing$ and $G[V(\gamma) \cup V(\gamma')] = \gamma\cup\gamma'$. Suppose that, for all polymers $\gamma\in\mathcal{C}$, the weight $w_\gamma$ can be computed in time $\exp(O(\abs{\gamma}))$. Then the truncated cluster expansion $T_m(Z(\mathcal{C},w))$ can be computed in time $\exp(O(m))\cdot\abs{G}^{O(1)}$.
\end{lemma}
\begin{proof}
    We can list all clusters of size at most $m$ in time $\exp(O(m))\cdot\abs{G}^{O(1)}$ by Lemma~\ref{lemma:ListClustersAlgorithm}. For each of these clusters, we can compute the Ursell function in time $\exp(O(m))$ by Lemma~\ref{lemma:UrsellFunctionAlgorithm}, and the polymer weights in time $\exp(O(m))$ by assumption. Hence, the truncated cluster expansion $T_m(Z(\mathcal{C},w))$ can be computed in time $\exp(O(m))\cdot\abs{G}^{O(1)}$.
\end{proof}

Combining Lemma~\ref{lemma:ConvergenceClusterExpansion} with Lemma~\ref{lemma:TruncatedClusterExpansionApproximationAlgorithm} proves Theorem~\ref{theorem:ApproximationAlgorithmAbstractPolymerModelPartitionFunction}.

\section{Applications}
\label{section:Applications}

In this section, we apply Theorem~\ref{theorem:ApproximationAlgorithmAbstractPolymerModelPartitionFunction} to establish an efficient algorithm for classes of counting problems in local lemma regimes. This includes the probability of intersection, the dimension of intersection for commuting projectors, and the dimension of intersection for general projectors. For the general case, we establish two algorithms: (1) under a global inclusion-exclusion stability condition, and (2) under a spectral gap assumption yielding an affine approximation.

\subsection{Probability of Intersection}
\label{section:ProbabilityOfIntersection}

In this section, we study the problem of approximating the probability of intersection. That is, given a set of events in a probability space, we aim to approximate the probability that all of them occur. Computing the probability of intersection is \#\textsc{P}-hard in general, and admits no fully polynomial-time randomised approximation scheme unless \textsc{RP} equals \textsc{NP}~\cite{dyer2004relative}. Nevertheless, we show that this problem admits an efficient approximation algorithm in a local lemma regime. As a corollary, we obtain an efficient algorithm for approximating the number of satisfying assignments of a CNF formula in a local lemma regime.

The Lov\'asz local lemma~\cite{erdos1975problems, shearer1985problem} provides a sufficient condition under which the probability of intersection is positive. Moser and Tardos~\cite{moser2010constructive} provided a constructive version of the lemma, establishing an efficient randomised algorithm for finding an assignment in the intersection of the events. The problem of approximately counting and sampling in local lemma regimes has received considerable attention. In the general setting, Barvinok~\cite{barvinok2025computing} established a quasi-polynomial time algorithm for approximately counting the intersection of events in a local lemma regime. In more restrictive settings, efficient approximate counting and sampling algorithms have been established for structured cases~\cite{moitra2019approximate, guo2019uniform, feng2021sampling, jain2022towards, jain2021sampling, he2021perfect, he2022sampling, feng2023towards, he2023deterministic, wang2024sampling}.

We consider a set of events $\{E_v\}_{v \in V(G)}$ indexed by the vertices of a finite graph $G$. The graph $G$ is called a \emph{strong dependency graph} of the events $\{E_v\}_{v \in V(G)}$ if, for every pair of disjoint vertex subsets $S,T \subseteq V(G)$ such that $G[S \cup T]=G[S] \cup G[T]$, the sets of events $\{E_v\}_{v \in S}$ and $\{E_v\}_{v \in T}$ are independent. Note that this notion of a dependency graph is stronger than the conventional definition used in the Lov\'asz local lemma. We are interested in the probability of intersection $\Pr\left[\bigcap_{v \in V(G)}E_v\right]$. Our algorithmic result concerning the approximation of $\Pr\left[\bigcap_{v \in V(G)}E_v\right]$ is as follows.

\begin{theorem}
    \label{theorem:ApproximationAlgorithmProbabilityIntersection}
    Fix $\delta>0$ and $\Delta\in\mathbb{Z}_{\geq2}$. Let $\{E_v\}_{v \in V(G)}$ be a set of events in a probability space with strong dependency graph $G$ of maximum degree $\Delta$ and chromatic number $\chi$. Suppose that, for all $v \in V(G)$,
    \begin{equation}
        \Pr\left[\overline{E_v}\right] \leq \left(\frac{1}{e^{1+\delta}(2\Delta+1)}\right)^\chi. \notag 
    \end{equation}
    Then the cluster expansion for $\log\left(\Pr\left[\bigcap_{v \in V(G)}E_v\right]\right)$ converges absolutely and the probability of $\bigcap_{v \in V(G)}E_v$ is positive. Further, if for all $U \subseteq V(G)$, $\Pr\left[\bigcap_{v \in U}\overline{E_v}\right]$ can be computed in time $\exp(O(\abs{U}))$, then there is a fully polynomial-time approximation scheme for the probability of $\bigcap_{v \in V(G)}E_v$.
\end{theorem}

Our algorithm is similar to that of Barvinok's~\cite{barvinok2025computing}, which is obtained via an interpolation-based approach~\cite{barvinok2016combinatorics}; however, ours proceeds directly via the cluster expansion. Barvinok's approach yields a quasi-polynomial time algorithm and further outlines how the method of Patel and Regts~\cite{patel2017deterministic} can be applied to obtain a polynomial-time algorithm. The resulting bound is more permissive than ours, allowing for significantly larger event probabilities. Barvinok also analysed the problem through a cluster expansion approach via the Koteck\'y-Preiss convergence
criterion~\cite{kotecky1986cluster}, obtaining a bound of the form $\Delta^{-O(\Delta)}$. Our bound recovers this condition once a bound on the chromatic number is applied.

We apply Theorem~\ref{theorem:ApproximationAlgorithmProbabilityIntersection} to the problem of counting satisfying assignments of $k$-CNF formulae. We consider a set of $k$-clauses $\{C_v\}_{v \in V(G)}$ over Boolean variables. In this context, each event corresponds to the satisfaction of a clause under a uniformly random assignment to the variables, and the strong dependency graph $G$ has an edge between vertices $u$ and $v$ if and only if the clauses $C_u$ and $C_v$ share a variable. We are interested in the probability of intersection $\Pr\left[\bigwedge_{v \in V(G)}C_v\right]$, which is proportional to the number of satisfying assignments up to an easily computed factor. Our algorithmic result concerning the approximation of $\Pr\left[\bigwedge_{v \in V(G)}C_v\right]$ is as follows.

\begin{corollary}
    Fix $\delta>0$ and $k,\Delta\in\mathbb{Z}_{\geq2}$. Let $\{C_v\}_{v \in V(G)}$ be a set of $k$-clauses with strong dependency graph $G$ of maximum degree $\Delta$ and chromatic number $\chi$. Suppose that 
    \begin{equation}
        k \geq \frac{\chi}{\log(2)}\left(\log(2\Delta+1)+1+\delta\right). \notag
    \end{equation}
    Then $\bigwedge_{v \in V(G)}C_v$ is satisfiable and there is a fully polynomial-time approximation scheme for counting the number of satisfying assignments.
\end{corollary}
\begin{proof}
    For any $v \in V(G)$, we have
    \begin{equation}
        \Pr[\lnot C_v] = \frac{1}{2^k} \leq \left(\frac{1}{e^{1+\delta}(2\Delta+1)}\right)^\chi. \notag
    \end{equation}
    Further, for any $U \subseteq V(G)$, $\Pr\left[\bigwedge_{v \in U}\lnot C_v\right]$ can be computed in time $\exp(O(\abs{U}))$ by direct enumeration. The proof then follows from Theorem~\ref{theorem:ApproximationAlgorithmProbabilityIntersection}.
\end{proof}

Previous work established approximate counting and sampling algorithms for satisfying assignments of CNF formulae under the condition that $k\geq\frac{1}{\log(2)}\left(\log(k)+5\log(\Delta+1)+O(1)\right)$~\cite{he2023deterministic}. Our bound exhibits the same logarithmic dependence on the dependency degree but scales linearly with the chromatic number, offering an improvement only when the dependency graph has small chromatic number. This logarithmic dependence on the dependency degree is asymptotically optimal, assuming \textsc{P} is not equal to \textsc{NP}~\cite{bezakova2019approximation, galanis2023inapproximability}.

We prove Theorem~\ref{theorem:ApproximationAlgorithmProbabilityIntersection} by showing that the conditions required to apply Theorem~\ref{theorem:ApproximationAlgorithmAbstractPolymerModelPartitionFunction} are satisfied. That is, we show that (1) the probability of intersection $\Pr\left[\bigcap_{v \in V(G)}E_v\right]$ admits a suitable abstract polymer model representation, and (2) the polymer weights satisfy the desired bound.

\begin{lemma}
    \label{lemma:ProbabilityIntersectionAbstractPolymerModel}
    The probability of intersection admits the following abstract polymer model representation.
    \begin{equation}
        \Pr\left[\bigcap_{v \in V(G)}E_v\right] = \sum_{\Gamma\in\mathcal{G}}\prod_{\gamma\in\Gamma}w_\gamma, \notag
    \end{equation}
    where
    \begin{equation}
        w_\gamma \coloneqq (-1)^\abs{\gamma}\Pr\left[\bigcap_{v \in V(\gamma)}\overline{E_v}\right]. \notag
    \end{equation}
\end{lemma}
\begin{proof}
    By De Morgan's law,
    \begin{equation}
        \Pr\left[\bigcap_{v \in V(G)}E_v\right] = 1-\Pr\left[\bigcup_{v \in V(G)}\overline{E_v}\right]. \notag
    \end{equation}
    Now, by the principle of inclusion-exclusion, we obtain
    \begin{equation}
        \Pr\left[\bigcup_{v \in V(G)}\overline{E_v}\right] = 1-\sum_{\substack{S \subseteq V(G) \\ S \neq \varnothing}}(-1)^{\abs{S}+1}\Pr\left[\bigcap_{v \in S}\overline{E_v}\right] = \sum_{S \subseteq V(G)}(-1)^\abs{S}\Pr\left[\bigcap_{v \in S}\overline{E_v}\right]. \notag
    \end{equation}
    For a subset $S \subseteq V(G)$, let $\Gamma_S$ denote the maximally connected components of the induced subgraph $G[S]$. By factorising over these components, we have
    \begin{align}
        \Pr\left[\bigcap_{v \in V(G)}E_v\right] &= \sum_{S \subseteq V}\prod_{\gamma\in\Gamma_S}(-1)^\abs{\gamma}\Pr\left[\bigcap_{v \in V(\gamma)}\overline{E_v}\right] \notag \\
        &= \sum_{\Gamma\in\mathcal{G}}\prod_{\gamma\in\Gamma}w_\gamma. \notag
    \end{align}
    This completes the proof.
\end{proof}

\begin{lemma}
    \label{lemma:ProbabilityIntersectionWeightBound}
    Fix $\delta>0$ and $\Delta\in\mathbb{Z}_{\geq2}$. Let $\{E_v\}_{v \in V(G)}$ be a set of events in a probability space with strong dependency graph $G$ of maximum degree $\Delta$ and chromatic number $\chi$. Suppose that, for all $v \in V(G)$, 
    \begin{equation}
        \Pr\left[\overline{E_v}\right] \leq \left(\frac{1}{e^{1+\delta}(2\Delta+1)}\right)^\chi. \notag 
    \end{equation}
    Then, for all polymers $\gamma\in\mathcal{C}$, the weight $w_\gamma$ satisfies
    \begin{equation}
        \abs{w_\gamma} \leq \left(\frac{1}{e^{1+\delta}(2\Delta+1)}\right)^\abs{\gamma}. \notag 
    \end{equation}
\end{lemma}
\begin{proof}
    For any polymer $\gamma\in\mathcal{C}$, we have
    \begin{equation}
        \abs{w_\gamma} = \Pr\left[\bigcap_{v \in V(\gamma)}\overline{E_v}\right] = \mathbb{E}\left[\prod_{v \in V(\gamma)}\mathds{1}_{\overline{E_v}}\right], \notag
    \end{equation}
    where $\mathds{1}_E$ denotes the indicator function of the event $E$. Let $\{C_i\}_{i\in[\chi]}$ be a proper colouring of $G$. Note that this induces a proper colouring of $\gamma$. By applying H\"older's inequality with respect to the colouring classes, we obtain
    \begin{equation}
        \abs{w_\gamma} \leq \prod_{i=1}^\chi\left(\mathbb{E}\left[\prod_{v \in V(\gamma) \cap C_i}\mathds{1}_{\overline{E_v}}\right]\right)^\frac{1}{\chi} = \left(\prod_{v \in V(\gamma)}\Pr\left[\overline{E_v}\right]\right)^\frac{1}{\chi} \leq \left(\frac{1}{e^{1+\delta}(2\Delta+1)}\right)^\abs{\gamma}, \notag
    \end{equation}
    completing the proof.
\end{proof}

Combining Theorem~\ref{theorem:ApproximationAlgorithmAbstractPolymerModelPartitionFunction} with Lemma~\ref{lemma:ProbabilityIntersectionAbstractPolymerModel} and Lemma~\ref{lemma:ProbabilityIntersectionWeightBound} proves Theorem~\ref{theorem:ApproximationAlgorithmProbabilityIntersection}.

\subsection{Commuting Dimension of Intersection}
\label{section:CommutingDimensionOfIntersection}

In this section, we study the problem of approximating the dimension of intersection for commuting projectors. That is, given a set of pairwise commuting projectors on a Hilbert space, we aim to approximate the dimension of the intersection of their kernels. Computing the dimension of intersection for commuting projectors is at least as hard as computing the probability of intersection, and is therefore \#\textsc{P}-hard in general, and admits no fully polynomial-time randomised approximation scheme unless \textsc{RP} equals \textsc{NP}~\cite{dyer2004relative}. Nevertheless, we show that this problem admits an efficient approximation algorithm in a local lemma regime. As a corollary, we obtain an efficient algorithm for approximating the dimension of satisfying assignments of a commuting quantum satisfiability formula in a local lemma regime.

The quantum Lov\'asz local lemma~\cite{ambainis2012quantum, sattath2016local, he2019quantum} provides a sufficient condition under which the dimension of intersection for commuting projectors is positive. Constructive versions of the lemma establish efficient algorithms for finding a quantum state in the intersection for commuting projectors~\cite{arad2013constructive, schwarz2013information}.

Let $\mathcal{H}$ be a Hilbert space that is a tensor product of local spaces. We consider a set of pairwise commuting projectors $\{\Pi_v\}_{v \in V(G)}$ on a Hilbert space indexed by the vertices of a finite graph $G$. The graph $G$ is called the \emph{strong dependency graph} of the projectors $\{\Pi_v\}_{v \in V(G)}$ if there is an edge between vertices $u$ and $v$ if and only if the supports of the projectors $\Pi_u$ and $\Pi_v$ are not disjoint. Note that this notion of a dependency graph is stronger than the conventional definition used in the quantum Lov\'asz local lemma. We denote the dimension, kernel, rank, image, and trace by $\dim[\;\cdot\;]$, $\ker[\;\cdot\;]$, $\rank[\;\cdot\;]$, $\operatorname{im}[\;\cdot\;]$, and $\tr[\;\cdot\;]$, respectively, with all quantities normalised by the dimension of the subspace on which the operator acts non-trivially. For convenience, we represent vector subspaces as kernels of projection operators. We are interested in the dimension of intersection $\dim\left[\bigcap_{v \in V(G)}\ker\Pi_v\right]$ for commuting projectors. Our algorithmic result concerning the approximation of $\dim\left[\bigcap_{v \in V(G)}\ker\Pi_v\right]$ is as follows.

\begin{theorem}
    \label{theorem:ApproximationAlgorithmDimensionIntersectionCommuting}
    Fix $\delta>0$ and $\Delta\in\mathbb{Z}_{\geq2}$. Let $\mathcal{H}$ be a Hilbert space that is a tensor product of local spaces. Further let $\{\Pi_v\}_{v \in V(G)}$ be a set of pairwise commuting projectors on $\mathcal{H}$ with strong dependency graph $G$ of maximum degree $\Delta$ and chromatic number $\chi$. Suppose that, for all $v \in V(G)$, 
    \begin{equation}
        \rank[\Pi_v] \leq \left(\frac{1}{e^{1+\delta}(2\Delta+1)}\right)^\chi. \notag 
    \end{equation}
    Then the cluster expansion for $\log\left(\dim\left[\bigcap_{v \in V(G)}\ker\Pi_v\right]\right)$ converges absolutely and the dimension of $\bigcap_{v \in V(G)}\ker\Pi_v$ is positive. Further, if for all $U \subseteq V(G)$, $\dim\left[\bigcap_{v \in U}\operatorname{im}\Pi_v\right]$ can be computed in time $\exp(O(\abs{U}))$, then there is a fully polynomial-time approximation scheme for the dimension of $\bigcap_{v \in V(G)}\ker\Pi_v$.
\end{theorem}

We apply Theorem~\ref{theorem:ApproximationAlgorithmDimensionIntersectionCommuting} to the problem of computing the dimension of the satisfying subspace of a commuting quantum $k$-satisfiability formula. We consider a set of pairwise commuting $k$-local projectors $\{\Pi_v\}_{v \in V(G)}$ on a Hilbert space that is a tensor product of $d$-dimensional local spaces. We are interested in the dimension of intersection $\dim\left[\bigcap_{v \in V(G)}\ker\Pi_v\right]$, which is exactly the dimension of the satisfying subspace. Our algorithmic result concerning the approximation of $\dim\left[\bigcap_{v \in V(G)}\ker\Pi_v\right]$ is as follows.

\begin{corollary}
    Fix $\delta>0$ and $d,k,\Delta\in\mathbb{Z}_{\geq2}$. Let $\mathcal{H}$ be a Hilbert space that is a tensor product of $d$-dimensional local spaces. Further let $\{\Pi_v\}_{v \in V(G)}$ be a set of pairwise commuting $k$-local projectors on $\mathcal{H}$ with strong dependency graph $G$ of maximum degree $\Delta$ and chromatic number $\chi$. Suppose that, for all $v \in V(G)$, 
    \begin{equation}
        \rank[\Pi_v] \leq \left(\frac{1}{e^{1+\delta}(2\Delta+1)}\right)^\chi. \notag 
    \end{equation}
    Then $\{\Pi_v\}_{v \in V(G)}$ is satisfiable and there is a fully polynomial-time approximation scheme for the dimension of $\bigcap_{v \in V(G)}\ker\Pi_v$.
\end{corollary}
\begin{proof}
    For any $U \subseteq V(G)$, $\dim\left[\bigcap_{v \in U}\operatorname{im}\Pi_v\right]$ can be computed in time $\exp(O(\abs{U}))$ by diagonalising the $\exp(O(\abs{U}))$-dimensional support. The proof then follows from Theorem~\ref{theorem:ApproximationAlgorithmDimensionIntersectionCommuting}.
\end{proof}

We prove Theorem~\ref{theorem:ApproximationAlgorithmDimensionIntersectionCommuting} by showing that the conditions required to apply Theorem~\ref{theorem:ApproximationAlgorithmAbstractPolymerModelPartitionFunction} are satisfied. That is, we show that (1) the dimension of intersection for commuting projectors $\dim\left[\bigcap_{v \in V(G)}\ker\Pi_v\right]$ admits a suitable abstract polymer model representation, and (2) the polymer weights satisfy the desired bound.

\begin{lemma}
    \label{lemma:DimensionIntersectionCommutingAbstractPolymerModel}
    The dimension of intersection for commuting projectors admits the following abstract polymer model representation.
    \begin{equation}
        \dim\left[\bigcap_{v \in V(G)}\ker\Pi_v\right] = \sum_{\Gamma\in\mathcal{G}}\prod_{\gamma\in\Gamma}w_\gamma, \notag
    \end{equation}
    where
    \begin{equation}
        w_\gamma \coloneqq (-1)^\abs{\gamma}\dim\left[\bigcap_{v \in V(\gamma)}\operatorname{im}\Pi_v\right]. \notag
    \end{equation}
\end{lemma}
\begin{proof}
     Since $\{\Pi_v\}_{v \in V(G)}$ are pairwise commuting projectors, the product $\prod_{v \in V(G)}(\mathbb{I}-\Pi_v)$ is the projector onto $\bigcap_{v \in V(G)}\ker\Pi_v$. Similarly, for any subset $S \subseteq V(G)$, the product $\prod_{v\in S}\Pi_v$ is the projector onto $\bigcap_{v \in S}\operatorname{im}\Pi_v$. Hence,
    \begin{align}
        \dim\left[\bigcap_{v \in V(G)}\ker\Pi_v\right] &= \tr\left[\prod_{v \in V(G)}(\mathbb{I}-\Pi_v)\right] \notag \\ 
        &= \sum_{S \subseteq V(G)}(-1)^\abs{S}\tr\left[\prod_{v \in S}\Pi_v\right] \notag \\ 
        &= \sum_{S \subseteq V(G)}(-1)^\abs{S}\dim\left[\bigcap_{v \in S}\operatorname{im}\Pi_v\right]. \notag
    \end{align}
    For a subset $S \subseteq V(G)$, let $\Gamma_S$ denote the maximally connected components of the induced subgraph $G[S]$. By factorising over these components, we have
    \begin{align}
        \dim\left[\bigcap_{v \in V(G)}\ker\Pi_v\right] &= \sum_{S \subseteq V}\prod_{\gamma\in\Gamma_S}(-1)^\abs{\gamma}\dim\left[\bigcap_{v \in V(\gamma)}\operatorname{im}\Pi_v\right] \notag \\
        &= \sum_{\Gamma\in\mathcal{G}}\prod_{\gamma\in\Gamma}w_\gamma. \notag
    \end{align}
    This completes the proof.
\end{proof}

\begin{lemma}
    \label{lemma:DimensionIntersectionCommutingWeightBound}
    Fix $\delta>0$ and $\Delta\in\mathbb{Z}_{\geq2}$. Let $\mathcal{H}$ be a Hilbert space that is a tensor product of local spaces. Further let $\{\Pi_v\}_{v \in V(G)}$ be a set of pairwise commuting projectors on $\mathcal{H}$ with strong dependency graph $G$ of maximum degree $\Delta$ and chromatic number $\chi$. Suppose that, for all $v \in V(G)$, 
    \begin{equation}
        \rank[\Pi_v] \leq \left(\frac{1}{e^{1+\delta}(2\Delta+1)}\right)^\chi. \notag 
    \end{equation}
    Then, for all polymers $\gamma\in\mathcal{C}$, the weight $w_\gamma$ satisfies
    \begin{equation}
        \abs{w_\gamma} \leq \left(\frac{1}{e^{1+\delta}(2\Delta+1)}\right)^\abs{\gamma}. \notag 
    \end{equation}
\end{lemma}
\begin{proof}
    For any polymer $\gamma\in\mathcal{C}$, we have
    \begin{equation}
        \abs{w_\gamma} = \frac{1}{d^\abs{\operatorname{supp}(\gamma)}}\norm{\prod_{v \in V(\gamma)}\Pi_v}_1, \notag
    \end{equation}
    where $\operatorname{supp}(\gamma)\coloneqq\bigcup_{v \in V(\gamma)}\operatorname{supp}(\Pi_v)$ and $\norm{\;\cdot\;}_p$ denotes the Schatten $p$-norm. Let $\{C_i\}_{i\in[\chi]}$ be a proper colouring of $G$. Note that this induces a proper colouring of $\gamma$. By applying H\"older's inequality with respect to the colouring classes, we obtain
    \begin{equation}
        \abs{w_\gamma} \leq \frac{1}{d^\abs{\operatorname{supp}(\gamma)}}\prod_{i=1}^\chi\norm{\prod_{v \in V(\gamma) \cap C_i}\Pi_v}_\chi \leq \left(\prod_{v \in V(\gamma)}\rank[\Pi_v]\right)^\frac{1}{\chi} \leq \left(\frac{1}{e^{1+\delta}(2\Delta+1)}\right)^\abs{\gamma}, \notag
    \end{equation}
    completing the proof.
\end{proof}

Combining Theorem~\ref{theorem:ApproximationAlgorithmAbstractPolymerModelPartitionFunction} with Lemma~\ref{lemma:DimensionIntersectionCommutingAbstractPolymerModel} and Lemma~\ref{lemma:DimensionIntersectionCommutingWeightBound} proves Theorem~\ref{theorem:ApproximationAlgorithmDimensionIntersectionCommuting}.

\subsection{General Dimension of Intersection}
\label{section:GeneralDimensionOfIntersection}

In this section, we study the problem of approximating the dimension of intersection for general projectors. That is, given a set of projectors on a Hilbert space, we aim to approximate the dimension of the intersection of their kernels. From a physics perspective, the dimension of intersection of projector kernels corresponds to the \emph{ground state degeneracy} of the associated Hamiltonian. Computing the dimension of intersection for general projectors is at least as hard as for commuting projectors, and is therefore \#\textsc{P}-hard in general, and admits no fully polynomial-time randomised approximation scheme unless \textsc{RP} equals \textsc{NP}~\cite{dyer2004relative}. Nevertheless, we show that this problem admits an efficient approximation algorithm in a local lemma regime. However, our approach requires a global inclusion-exclusion stability condition, since the standard inclusion-exclusion formula no longer holds for general projectors. As a corollary, we obtain an efficient algorithm for approximating the dimension of satisfying assignments of a quantum satisfiability formula in a local lemma regime. Further, we show that the global stability requirement can be removed under the assumption of a spectral gap, yielding an affine approximation to the dimension of intersection.

The quantum Lov\'asz local lemma~\cite{ambainis2012quantum, sattath2016local, he2019quantum} provides a sufficient condition under which the dimension of intersection for general projectors is positive. A constructive version of the lemma has established an efficient algorithm for finding a quantum state in the intersection for general projectors under the assumption of a uniform spectral gap~\cite{gilyen2017preparing}.

Let $\mathcal{H}$ be a Hilbert space that is a tensor product of local spaces. We consider a set of projectors $\{\Pi_v\}_{v \in V(G)}$ on a Hilbert space indexed by the vertices of a finite graph $G$. As in the commuting setting, the graph $G$ is called the \emph{strong dependency graph} of the projectors $\{\Pi_v\}_{v \in V(G)}$ if there is an edge between vertices $u$ and $v$ if and only if the supports of the projectors $\Pi_u$ and $\Pi_v$ are not disjoint. We retain the same conventions as in the commuting case for dimension, kernel, rank, image, and trace, and denote the codimension by $\operatorname{codim}[\;\cdot\;]\coloneqq1-\dim[\;\cdot\;]$. We are interested in the dimension of intersection $\dim\left[\bigcap_{v \in V(G)}\ker\Pi_v\right]$ for general projectors. Our algorithmic result concerning the approximation of $\dim\left[\bigcap_{v \in V(G)}\ker\Pi_v\right]$ is as follows.

\begin{theorem}
    \label{theorem:ApproximationAlgorithmDimensionIntersectionGeneral}
    Fix $\delta>0$ and $\Delta\in\mathbb{Z}_{\geq2}$. Let $\mathcal{H}$ be a Hilbert space that is a tensor product of local spaces. Further let $\{\Pi_v\}_{v \in V(G)}$ be a set of projectors on $\mathcal{H}$ with strong dependency graph $G$ of maximum degree $\Delta$. Suppose that, for all $U \subseteq V(G)$, 
    \begin{equation}
        \abs{\sum_{S \subseteq U}(-1)^\abs{S}\dim\left[\bigcap_{v \in S}\ker\Pi_v\right]} \leq \left(\frac{1}{e^{1+\delta}(2\Delta+1)}\right)^\abs{U}. \notag 
    \end{equation}
    Then the cluster expansion for $\log\left(\dim\left[\bigcap_{v \in V(G)}\ker\Pi_v\right]\right)$ converges absolutely and the dimension of $\bigcap_{v \in V(G)}\ker\Pi_v$ is positive. Further, if for all $U \subseteq V(G)$, $\dim\left[\bigcap_{v \in U}\ker\Pi_v\right]$ can be computed in time $\exp(O(\abs{U}))$, then there is a fully polynomial-time approximation scheme for the dimension of $\bigcap_{v \in V(G)}\ker\Pi_v$.
\end{theorem}

We apply Theorem~\ref{theorem:ApproximationAlgorithmDimensionIntersectionGeneral} to the problem of computing the dimension of the satisfying subspace of a quantum $k$-satisfiability formula. We consider a set of $k$-local projectors $\{\Pi_v\}_{v \in V(G)}$ on a Hilbert space that is a tensor product of $d$-dimensional local spaces. We are interested in the dimension of intersection $\dim\left[\bigcap_{v \in V(G)}\ker\Pi_v\right]$, which is exactly the dimension of the satisfying subspace. Our algorithmic result concerning the approximation of $\dim\left[\bigcap_{v \in V(G)}\ker\Pi_v\right]$ is as follows.

\begin{corollary}
    Fix $\delta>0$ and $d,k,\Delta\in\mathbb{Z}_{\geq2}$. Let $\mathcal{H}$ be a Hilbert space that is a tensor product of $d$-dimensional local spaces. Further let $\{\Pi_v\}_{v \in V(G)}$ be a set of $k$-local projectors on $\mathcal{H}$ with strong dependency graph $G$ of maximum degree $\Delta$. Suppose that, for all $U \subseteq V(G)$, 
    \begin{equation}
        \abs{\sum_{S \subseteq U}(-1)^\abs{S}\dim\left[\bigcap_{v \in S}\ker\Pi_v\right]} \leq \left(\frac{1}{e^{1+\delta}(2\Delta+1)}\right)^\abs{U}. \notag 
    \end{equation}
    Then $\{\Pi_v\}_{v \in V(G)}$ is satisfiable and there is a fully polynomial-time approximation scheme for the dimension of $\bigcap_{v \in V(G)}\ker\Pi_v$.
\end{corollary}
\begin{proof}
    For any $U \subseteq V(G)$, $\dim\left[\bigcap_{v \in U}\ker\Pi_v\right]$ can be computed in time $\exp(O(\abs{U}))$ by diagonalising the $\exp(O(\abs{U}))$-dimensional support. The proof then follows from Theorem~\ref{theorem:ApproximationAlgorithmDimensionIntersectionGeneral}.
\end{proof}

In the case $d=2$ and $k=2$, a complete characterisation of the satisfying subspace in the quantum Lov\'asz local lemma regime is known, via an explicit description as the span of tree tensor network states~\cite{ji2011complete}. However, computing the dimension of this subspace is \#\textsc{P}-complete~\cite{ji2011complete}.

We prove Theorem~\ref{theorem:ApproximationAlgorithmDimensionIntersectionGeneral} by showing that the conditions required to apply Theorem~\ref{theorem:ApproximationAlgorithmAbstractPolymerModelPartitionFunction} are satisfied. That is, we show that the dimension of intersection for general projectors $\dim\left[\bigcap_{v \in V(G)}\ker\Pi_v\right]$ admits a suitable abstract polymer model representation.

\begin{lemma}
    \label{lemma:DimensionIntersectionGeneralAbstractPolymerModel}
    The dimension of intersection for general projectors admits the following abstract polymer model representation.
    \begin{equation}
        \dim\left[\bigcap_{v \in V(G)}\ker\Pi_v\right] = \sum_{\Gamma\in\mathcal{G}}\prod_{\gamma\in\Gamma}w_\gamma, \notag
    \end{equation}
    where
    \begin{equation}
        w_\gamma \coloneqq (-1)^\abs{\gamma}\sum_{T \subseteq V(\gamma)}(-1)^\abs{T}\dim\left[\bigcap_{v \in T}\ker\Pi_v\right]. \notag
    \end{equation}
\end{lemma}
\begin{proof}
    By the principle of inclusion-exclusion (see for example \cite[Theorem 12.1]{graham1995handbook}),
    \begin{equation}
        \dim\left[\bigcap_{v \in V(G)}\ker\Pi_v\right] = \sum_{S \subseteq V(G)}(-1)^\abs{S}\sum_{T \subseteq S}(-1)^\abs{T}\dim\left[\bigcap_{v \in T}\ker\Pi_v\right]. \notag
    \end{equation}
    For a subset $S \subseteq V(G)$, let $\Gamma_S$ denote the maximally connected components of the induced subgraph $G[S]$. By factorising over these components, we have
    \begin{align}
        \dim\left[\bigcap_{v \in V(G)}\ker\Pi_v\right] &= \sum_{S \subseteq V(G)}\prod_{\gamma\in\Gamma_S}(-1)^\abs{\gamma}\sum_{T \subseteq V(\gamma)}(-1)^\abs{T}\dim\left[\bigcap_{v \in T}\ker\Pi_v\right] \notag \\
        &= \sum_{\Gamma\in\mathcal{G}}\prod_{\gamma\in\Gamma}w_\gamma. \notag
    \end{align}
    This completes the proof.
\end{proof}

Combining Theorem~\ref{theorem:ApproximationAlgorithmAbstractPolymerModelPartitionFunction} with Lemma~\ref{lemma:DimensionIntersectionGeneralAbstractPolymerModel} proves Theorem~\ref{theorem:ApproximationAlgorithmDimensionIntersectionGeneral}. We now establish a result which does not require a global inclusion-exclusion stability condition under the assumption of a spectral gap; however, it yields only an affine approximation to the dimension of intersection.

\begin{theorem}
    \label{theorem:DetectabilityApproximationAlgorithmDimensionIntersectionGeneral}
    Fix $\delta,\epsilon>0$, $T\in\mathbb{Z}^+$, and $d,k,\Delta\in\mathbb{Z}_{\geq2}$. Let $\mathcal{H}$ be a Hilbert space that is a tensor product of $d$-dimensional local spaces. Further let $\{\Pi_v\}_{v \in V(G)}$ be a set of $k$-local projectors on $\mathcal{H}$ with strong dependency graph $G$ of maximum degree $\Delta$ and chromatic number $\chi$. Finally let $\lambda_\star\geq0$ denote the spectral gap of $\sum_{v \in V(G)}\Pi_v$. Suppose that, for all $v \in V(G)$, 
    \begin{equation}
        \rank[\Pi_v] \leq \left(\frac{1}{e^{1+\delta}(2T(\Delta+1)-1)}\right)^{T\chi}. \notag 
    \end{equation}
    Then there is a polynomial-time algorithm that outputs a number $z$ such that
    \begin{equation}
        \abs{z-\dim\left[\bigcap_{v \in V(G)}\ker\Pi_v\right]} \leq \notag \epsilon\dim\left[\bigcap_{v \in V(G)}\ker\Pi_v\right] + (1+\epsilon)\left(\frac{1}{1+\frac{\lambda_\star}{\chi^2
        }}\right)^\frac{T}{2}.
    \end{equation}
\end{theorem}

We prove Theorem~\ref{theorem:DetectabilityApproximationAlgorithmDimensionIntersectionGeneral} by first considering the \emph{dimension of detectability}, which provides an approximation to the dimension of intersection via a finite sequence of projections. We then show that the conditions required to apply Theorem~\ref{theorem:ApproximationAlgorithmAbstractPolymerModelPartitionFunction} are satisfied. That is, we show that (1) the dimension of detectability admits a suitable abstract polymer model representation, and (2) the polymer weights satisfy the desired bound. Finally, we use the detectability lemma~\cite{aharonov2009detectability, anshu2016simple} to bound the difference between the dimension of detectability and the dimension of intersection.

Let $T$ be a positive integer and $\{C_i\}_{i\in[\chi]}$ a proper colouring of the strong dependency graph $G$. We are interested in the dimension of detectability $\tr\left[\left(\prod_{i=1}^\chi\prod_{v \in V(G) \cap C_i}(\mathbb{I}-\Pi_v)\right)^T\right]$. For the dimension of detectability, we consider polymers that are connected induced subgraphs of $G \boxtimes K_T$. Further, all products of operators are taken with respect to a fixed ordering of the vertices.

\begin{lemma}
    \label{lemma:DimensionDetectabilityAbstractPolymerModel}
    The dimension of detectability admits the following abstract polymer model representation.
    \begin{equation}
        \tr\left[\left(\prod_{i=1}^\chi\prod_{v \in V(G) \cap C_i}(\mathbb{I}-\Pi_v)\right)^T\right] = \sum_{\Gamma\in\mathcal{G}}\prod_{\gamma\in\Gamma}w_\gamma, \notag
    \end{equation}
    where
    \begin{equation}
        w_\gamma \coloneqq (-1)^\abs{\gamma}\tr\left[\prod_{v \in V(\gamma)}\Pi_v\right]. \notag
    \end{equation}
\end{lemma}
\begin{proof}
    We impose a lexicographic ordering on the vertices of $G \boxtimes K_T$ by identifying each vertex with a triple $(\tau,i,v)$ with $\tau \in [T]$, $i \in [\chi]$, and $v \in C_i$. We shall take all products of operators with respect to this ordering. By expanding the trace, we obtain
    \begin{align}
        \tr\left[\left(\prod_{i=1}^\chi\prod_{v \in V(G) \cap C_i}(\mathbb{I}-\Pi_v)\right)^T\right] &= \tr\left[\prod_{v \in V(G \boxtimes K_T)}(\mathbb{I}-\Pi_v)\right] \notag \\
        &= \sum_{S \subseteq V(G \boxtimes K_T)}(-1)^\abs{S}\tr\left[\prod_{v \in S}\Pi_v\right]. \notag
    \end{align}
    For a subset $S \subseteq V(G)$, let $\Gamma_S$ denote the maximally connected components of the induced subgraph $(G \boxtimes K_T)[S]$. By factorising over these components, we have
    \begin{align}
        \tr\left[\left(\prod_{i=1}^\chi\prod_{v \in V(G) \cap C_i}(\mathbb{I}-\Pi_v)\right)^T\right] &= \sum_{S \subseteq V(G \boxtimes K_T)}\prod_{\gamma\in\Gamma_S}(-1)^\abs{\gamma}\tr\left[\prod_{v \in V(\gamma)}\Pi_v\right] \notag \\
        &= \sum_{\Gamma\in\mathcal{G}}\prod_{\gamma\in\Gamma}w_\gamma. \notag
    \end{align}
    This completes the proof.
\end{proof}

\begin{lemma}
    \label{lemma:DimensionDetectabilityWeightBound}
    Fix $\delta>0$, $T\in\mathbb{Z}^+$, and $\Delta\in\mathbb{Z}_{\geq2}$. Let $\mathcal{H}$ be a Hilbert space that is a tensor product of local spaces. Further let $\{\Pi_v\}_{v \in V(G)}$ be a set of projectors on $\mathcal{H}$ with strong dependency graph $G$ of maximum degree $\Delta$ and chromatic number $\chi$.
    Suppose that, for all $v \in V(G)$, 
    \begin{equation}
        \rank[\Pi_v] \leq \left(\frac{1}{e^{1+\delta}(2T(\Delta+1)-1)}\right)^{T\chi}. \notag 
    \end{equation}
    Then, for all polymers $\gamma\in\mathcal{C}$, the weight $w_\gamma$ satisfies
    \begin{equation}
        \abs{w_\gamma} \leq \left(\frac{1}{e^{1+\delta}(2T(\Delta+1)-1)}\right)^\abs{\gamma}. \notag 
    \end{equation}
\end{lemma}
\begin{proof}
    For any polymer $\gamma\in\mathcal{C}$, we have
    \begin{equation}
        \abs{w_\gamma} = \frac{1}{d^\abs{\operatorname{supp}(\gamma)}}\norm{\prod_{v \in V(\gamma)}\Pi_v}_1, \notag
    \end{equation}
    where $\operatorname{supp}(\gamma)\coloneqq\bigcup_{v \in V(\gamma)}\operatorname{supp}(\Pi_v)$ and $\norm{\;\cdot\;}_p$ denotes the Schatten $p$-norm. Let $\{C_i\}_{i\in[\chi]}$ be a proper colouring of $G$. We define a proper colouring of $G \boxtimes K_T$ by $\{C_{\tau,i}\}_{\tau\in[T],i\in[\chi]}$, where $C_{\tau,i}\coloneqq\{(\tau,i,v) \mid v \in C_i\}$. Note that this induces a proper colouring of $\gamma$. By applying H\"older's inequality with respect to the colouring classes, we obtain
    \begin{equation}
        \abs{w_\gamma} \leq \frac{1}{d^\abs{\operatorname{supp}(\gamma)}}\prod_{\tau=1}^{T}\prod_{i=1}^{\chi}\norm{\prod_{v \in V(\gamma) \cap C_{\tau,i}}\Pi_v}_{T\chi} \leq \left(\prod_{v \in V(\gamma)}\rank[\Pi_v]\right)^\frac{1}{T\chi} \leq \left(\frac{1}{e^{1+\delta}(2T(\Delta+1)-1)}\right)^\abs{\gamma}, \notag
    \end{equation}
    completing the proof.
\end{proof}

\begin{lemma}
    \label{lemma:DimensionDetectabilityIntersectionApproximationBound}
    Fix $\Delta\in\mathbb{Z}_{\geq2}$. Let $\mathcal{H}$ be a Hilbert space that is a tensor product of local spaces. Further let $\{\Pi_v\}_{v \in V(G)}$ be a set of projectors on $\mathcal{H}$ with strong dependency graph $G$ of maximum degree $\Delta$ and chromatic number $\chi$. Finally let $\lambda_\star\geq0$ denote the spectral gap of $\sum_{v \in V(G)}\Pi_v$. Suppose that, for all $v \in V(G)$,
    \begin{equation}
        \rank[\Pi_v] \leq \frac{1}{e(\Delta+1)}. \notag
    \end{equation}
    Then, for any $T\in\mathbb{Z}^+$ and any proper colouring $\{C_i\}_{i\in[\chi]}$ of $G$,
    \begin{equation}
        \abs{\tr\left[\left(\prod_{i=1}^\chi\prod_{v \in V(G) \cap C_i}(\mathbb{I}-\Pi_v)\right)^T\right]-\dim\left[\bigcap_{v \in V(G)}\ker\Pi_v\right]} \leq \left(\frac{1}{1+\frac{\lambda_\star}{\chi^2
        }}\right)^\frac{T}{2}\operatorname{codim}\left[\bigcap_{v \in V(G)}\ker\Pi_v\right]. \notag
    \end{equation}
\end{lemma}
\begin{proof}
    For convenience, we define the operator $\Pi_\star$ by
    \begin{equation}
        \Pi_\star \coloneqq \prod_{i=1}^\chi\prod_{v \in V(G) \cap C_i}(\mathbb{I}-\Pi_v). \notag
    \end{equation}
    We further let $\Pi_0$ denote the projector onto $\bigcap_{v \in V(G)}\ker\Pi_v$. Then, for any $v \in V(G)$, we have $(\mathbb{I}-\Pi_v)\Pi_0=\Pi_0$, and so $\Pi_\star\Pi_0=\Pi_0$. Thus, for any $T\in\mathbb{Z}^+$,
    \begin{equation}
        \dim\left[\bigcap_{v \in V(G)}\ker\Pi_v\right] = \tr\left[\Pi_0\right] = \tr\left[\Pi_\star^T\Pi_0\right]. \notag
    \end{equation}
    Hence,
    \begin{align}
        \abs{\tr\left[\left(\prod_{i=1}^\chi\prod_{v \in V(G) \cap C_i}(\mathbb{I}-\Pi_v)\right)^T\right]-\dim\left[\bigcap_{v \in V(G)}\ker\Pi_v\right]} &= \abs{\tr\left[\Pi_\star^T(\mathbb{I}-\Pi_0)\right]} \notag \\
        &\leq \norm{\Pi_\star^T(\mathbb{I}-\Pi_0)}\operatorname{codim}\left[\bigcap_{v \in V(G)}\ker\Pi_v\right] \notag \\
        &\leq \norm{\Pi_\star(\mathbb{I}-\Pi_0)}^T\operatorname{codim}\left[\bigcap_{v \in V(G)}\ker\Pi_v\right]\!, \notag
    \end{align}
    where $\norm{\;\cdot\;}$ denotes the operator norm. By the quantum Lov\'asz local lemma~\cite[Theorem 3]{ambainis2012quantum}, the assumption on the rank ensures that the dimension of $\bigcap_{v \in V(G)}\ker\Pi_v$ is positive. Then, as a direct corollary of the detectability lemma~\cite[Corollary~1]{anshu2016simple}, we have
    \begin{equation}
        \norm{\Pi_\star(\mathbb{I}-\Pi_0)} \leq \left(\frac{1}{1+\frac{\lambda_\star}{\chi^2
        }}\right)^\frac{1}{2}. \notag
    \end{equation}
    Thus,
    \begin{equation}
        \abs{\tr\left[\left(\prod_{i=1}^\chi\prod_{v \in V(G) \cap C_i}(\mathbb{I}-\Pi_v)\right)^T\right]-\dim\left[\bigcap_{v \in V(G)}\ker\Pi_v\right]} \leq \left(\frac{1}{1+\frac{\lambda_\star}{\chi^2
        }}\right)^\frac{T}{2}\operatorname{codim}\left[\bigcap_{v \in V(G)}\ker\Pi_v\right], \notag
    \end{equation}
    completing the proof.
\end{proof}

We now prove Theorem~\ref{theorem:DetectabilityApproximationAlgorithmDimensionIntersectionGeneral}.

\begin{proof}[Proof of Theorem~\ref*{theorem:DetectabilityApproximationAlgorithmDimensionIntersectionGeneral}]
    By Lemma~\ref{lemma:DimensionDetectabilityAbstractPolymerModel}, the dimension of detectability admits a polymer model representation where the polymers are connected induced subgraphs of $G \boxtimes K_T$ and that two polymers $\gamma$ and $\gamma'$ are compatible if and only if $(G \boxtimes K_T)[V(\gamma) \cup V(\gamma')] = \gamma\cup\gamma'$. Then, by Lemma~\ref{lemma:DimensionDetectabilityWeightBound}, the assumption on the rank ensures that, for all polymers $\gamma\in\mathcal{C}$, the weight $w_\gamma$ satisfies
    \begin{equation}
        \abs{w_\gamma} \leq \left(\frac{1}{e^{1+\delta}(2T(\Delta+1)-1)}\right)^\abs{\gamma}. \notag 
    \end{equation}
    Further, for all polymers $\gamma\in\mathcal{C}$, the weight $w_\gamma$ can be computed in time $\exp(O(\abs{\gamma}))$ by diagonalising the $\exp(O(\abs{\gamma}))$-dimensional support. Observe that the maximum degree of $G \boxtimes K_T$ is $T(\Delta+1)-1$. Now, by Theorem~\ref{theorem:ApproximationAlgorithmAbstractPolymerModelPartitionFunction}, there is a fully polynomial-time approximation scheme for the dimension of detectability. That is, for any $\epsilon>0$, there is a polynomial-time algorithm that outputs a number $z$ such that
    \begin{equation}
        \abs{z-\tr\left[\left(\prod_{i=1}^\chi\prod_{v \in V(G) \cap C_i}(\mathbb{I}-\Pi_v)\right)^T\right]} \leq \epsilon\tr\left[\left(\prod_{i=1}^\chi\prod_{v \in V(G) \cap C_i}(\mathbb{I}-\Pi_v)\right)^T\right]. \notag
    \end{equation}
    The assumption on the rank ensures that, for all $v \in V(G)$,
    \begin{equation}
        \rank[\Pi_v] \leq \frac{1}{e(\Delta+1)}. \notag
    \end{equation}
    Hence, by Lemma~\ref{lemma:DimensionDetectabilityIntersectionApproximationBound}, we have
    \begin{align}
        \abs{\tr\left[\left(\prod_{i=1}^\chi\prod_{v \in V(G) \cap C_i}(\mathbb{I}-\Pi_v)\right)^T\right]-\dim\left[\bigcap_{v \in V(G)}\ker\Pi_v\right]} &\leq \left(\frac{1}{1+\frac{\lambda_\star}{\chi^2
        }}\right)^\frac{T}{2}\operatorname{codim}\left[\bigcap_{v \in V(G)}\ker\Pi_v\right] \notag \\
        &\leq \left(\frac{1}{1+\frac{\lambda_\star}{\chi^2
        }}\right)^\frac{T}{2}. \notag
    \end{align}
    Therefore,
    \begin{align}
        \abs{z-\tr\left[\left(\prod_{i=1}^\chi\prod_{v \in V(G) \cap C_i}(\mathbb{I}-\Pi_v)\right)^T\right]} &\leq \epsilon\left(\dim\left[\bigcap_{v \in V(G)}\ker\Pi_v\right]+\left(\frac{1}{1+\frac{\lambda_\star}{\chi^2
        }}\right)^\frac{T}{2}\right). \notag
    \end{align}
    Then, by the triangle inequality,
    \begin{align}
        \abs{z-\dim\left[\bigcap_{v \in V(G)}\ker\Pi_v\right]} &\leq \epsilon\left(\dim\left[\bigcap_{v \in V(G)}\ker\Pi_v\right]+\left(\frac{1}{1+\frac{\lambda_\star}{\chi^2
        }}\right)^\frac{T}{2}\right)+\left(\frac{1}{1+\frac{\lambda_\star}{\chi^2
        }}\right)^\frac{T}{2} \notag \\
        &= \epsilon\dim\left[\bigcap_{v \in V(G)}\ker\Pi_v\right]+(1+\epsilon)\left(\frac{1}{1+\frac{\lambda_\star}{\chi^2
        }}\right)^\frac{T}{2}. \notag
    \end{align}
    This completes the proof.
\end{proof}

\section{Conclusion \& Outlook}
\label{section:ConclusionAndOutlook}

We have established efficient approximate counting algorithms for several natural problems in local lemma regimes. We obtained fully polynomial-time approximation schemes for both the probability of intersection and the dimension of intersection for commuting projectors. For general projectors, we provided two algorithms: a fully polynomial-time approximation scheme under a global inclusion-exclusion stability condition, and an efficient affine approximation under a spectral gap assumption. As corollaries of our results, we obtained efficient algorithms for approximating the number of satisfying assignments of CNF formulae and the dimension of satisfying subspaces of quantum satisfiability formulae in local lemma regimes.

Several important problems remain open. A natural question is whether the global inclusion-exclusion stability condition can be relaxed to obtain a fully polynomial-time approximation scheme for the dimension of intersection of general projectors under local conditions alone. Another important problem is the development of efficient approximate sampling algorithms in local lemma regimes, particularly for the quantum case, where no efficient algorithms are known. Finally, understanding the optimal dependence on parameters such as the maximum degree and chromatic number in our approximation schemes remains an interesting direction for future research.

\section*{Acknowledgements}

We thank Tyler Helmuth, Richard Jozsa, Youming Qiao, Guus Regts, and Sergii Strelchuk for helpful discussions. RLM was supported by the ARC Centre of Excellence for Quantum Computation and Communication Technology (CQC2T), project number CE170100012. GW was supported by a scholarship from the Sydney Quantum Academy, and the ARC Centre of Excellence for Quantum Computation and Communication Technology (CQC2T), project number CE170100012.

\bibliography{bibliography}

\end{document}